\documentclass[journal,doublecolumn,10pt]{IEEEtran}
\usepackage{epsfig,latexsym}
\usepackage{float}
 \usepackage{enumerate}
\usepackage{indentfirst}
\usepackage{amsmath}
\usepackage{amssymb}
\usepackage{times}
\usepackage{subfigure}
\usepackage{cite}
\usepackage{url}
\usepackage{color}
\usepackage{graphicx}
\usepackage{epstopdf}
\definecolor{MyDarkBlue}{rgb}{0,0.08,0.45}
\definecolor{yellow}{rgb}{0.99,0.99,0.70}
\definecolor{white}{rgb}{0.796,0.948,0.816}
\definecolor{black}{rgb}{0.00,0.00,0.00}

\def\diag{\textrm{diag}}

\newtheorem{Proposition}{Proposition}

\newtheorem{Lemma}{Lemma}

\newtheorem{Remark}{Remark}
\newtheorem{Corollary}{Corollary}

\newtheorem{theorem}{$\mathbf{Theorem}$}

\newtheorem{proof}{Proof}

\begin{document}
\title{\huge Reconfigurable Intelligent Surfaces Assisted Communications with Discrete Phase Shifts: How Many Quantization Levels\\ are Required to Achieve Full Diversity?}

\author{Peng Xu,  \IEEEmembership{Member, IEEE},
Gaojie Chen,  \IEEEmembership{Senior Member, IEEE}, Zheng Yang,  \IEEEmembership{Member, IEEE}, \\
and Marco Di Renzo, \IEEEmembership{Fellow, IEEE}

\thanks{
P. Xu is  with  Chongqing Key Laboratory of Mobile Communications Technology,  School of   Communication and Information Engineering,  Chongqing University of Posts and Telecommunications, Chongqing, 400065,  China.

G. Chen is with School of Engineering, University of Leicester, Leicester LE1 7RH, U.K. (gaojie.chen@leicester.ac.uk).

Z. Yang is with Fujian Provincial Engineering Technology Research Center of Photoelectric Sensing Application, Key Laboratory of OptoElectronic Science and Technology for
Medicine of Ministry of Education, Fujian Normal University, Fuzhou 350007, China (e-mail: zyfjnu@163.com).

M. Di Renzo is with Universit\'e Paris-Saclay, CNRS and CentraleSup\'elec, Laboratoire des Signaux et Syst\`emes,  Gif-sur-Yvette, France. (e-mail: marco.direnzo@centralesupelec.fr).
}\vspace{-2em}}
\maketitle
\begin{abstract}
Due to hardware limitations, the phase shifts of the reflecting elements of
reconfigurable intelligent surfaces (RISs)
 need to be quantized into discrete values.
This letter aims to unveil  the minimum required number of phase quantization levels $L$
 in order to achieve the full diversity order
in RIS-assisted wireless
communication systems.
 With the aid of an upper bound of the outage
probability, we first prove that the full diversity order is achievable provided that  $L$ is not less than three.
 If $L=2$, on the other hand, we prove that the achievable diversity order cannot exceed $(N+1)/2$, where $N$ is the number of reflecting  elements. This is obtained with the aid of a lower bound of the outage probability. Therefore, we prove that the minimum required value of $L$ to achieve the full diversity order is $L=3$. Simulation results verify the theoretical analysis and the impact of phase quantization levels on RIS-assisted communication systems.
\end{abstract}
\vspace{-0.5em}
\begin{IEEEkeywords}
Reconfigurable intelligent  surface, discrete phase shifts, phase errors, outage probability, diversity order.
\end{IEEEkeywords}
\vspace{-1em}
\section{Introduction}
Reconfigurable intelligent surfaces (RISs) have recently  received significant attention  due to their ability to intelligently reconfigure the wireless  environment by integrating passive reflectors on  flat surfaces \cite{IRS_Marco_Eurasip,IRS_Marco_Smart,huang2020holog}. Precisely, each element of the RIS can control the amplitude and phase of the reflected signal, such that multiple reflected signals can be co-phased through  passive beamforming. Compared with  other technologies, such as relaying, RISs do not need to use radio-frequency chains and enjoys
 a convenient and low-cost deployment \cite{IRS_Survey,IRS_Marco_Relay}.

Due to their potential benefits in the context of wireless communications, RISs have recently been employed for several applications, which include beamforming \cite{wu2019intell,wu2020beamforming},
 non-orthogonal multiple access \cite{IRS_NOMA_Ding,IRS_NOMA_Zheng}, physical layer security \cite{IRS_Security_Chu}, relaying \cite{IRS_Relay_Abdul,IRS_Marco_Relay}, etc. These existing works have demonstrated that a suitable design of the phase shifts
of the reflecting elements is necessary  to reap the advantages of RIS-aided transmission. However, most of these existing works assume that the phase shifts of the reflectors are continuous variables, which is difficult to implement in practice due to hardware limitations. Motivated by these considerations,
 recent  works have investigated the performance of RIS-assisted systems in the presence of
  quantization phase errors (e.g., \cite{IRS_NOMA_Zheng,wu2020beamforming,IRS_Phase_Error,zhang2020reconf,IRS_Error_Capacity,IRS_Marco}). The works in  \cite{wu2020beamforming} and \cite{IRS_NOMA_Zheng} solved optimization problems with respect to the discrete phase shifts. The works in \cite{zhang2020reconf,IRS_Phase_Error,IRS_Error_Capacity,IRS_Marco} analyzed the average signal-to-noise-ratio (SNR) and the achievable data rate.
{The works in \cite{zhang2020reconf,IRS_Phase_Error,IRS_Error_Capacity,IRS_Marco}
   demonstrated that the average SNR scales with $N^2$ for  large values of $N$,  where
   $N$ denotes the number of reflecting elements.
 However,
the exact relation between the diversity order and the number of phase quantization
 levels $L$ for arbitrary $N\geq 2$ was not  revealed. Motivated by these considerations,  we
  investigate the diversity order of RIS-assisted transmission in the presence
  of  quantization phase errors for any $N \ge 2$.}

More specifically, the  main contribution of this letter consists of unveiling the  minimum number of
 required phase quantization levels to achieve the full diversity order  $N$. To obtain this goal,  the full diversity order of  $N$ is first proved to be achievable if  $L\geq 3$, by using mathematical induction based on
  an upper bound of the outage probability. Then, the achievable diversity order is proved not to exceed  $\frac{N+1}2$ if $L=2$, with the aid of a lower bound of the outage probability conditioned on the event that each phase quantization error is close to the quantization boundary. Numerical results also demonstrate that the loss of the outage performance is negligible if $L\geq 3$, which provides  important  insights for the design of phase quantization in RIS-assisted systems.

\vspace{-1em}
\section{System Model and Preliminaries}\label{section_model}
We consider an RIS-aided transmission system, which consists of a single-antenna source ($S$), a single-antenna destination ($D$) and an RIS with $N$ reflecting elements.
The direct link between $S$ and $D$ is assumed to be weak and it is not considered\footnote{We  consider the direct link between $S$ and $D$ in Section \ref{section_discussion}.}.
The channel coefficients from $S$ to the RIS and from the RIS to $D$ are denoted by the column-vectors $\mathbf{h}_{SI}$ and $\mathbf{h}_{ID}$, respectively.
The $n$th elements of $\mathbf{h}_{SI}$ and $\mathbf{h}_{ID}$ are
$[\mathbf{h}_{SI}]_n\sim\mathcal{CN}(0,\Omega_S)$ and $[\mathbf{h}_{ID}]_n\sim\mathcal{CN}(0,\Omega_I)$, respectively,
$\forall n\in\{1,\ldots,N\}$, {where $\mathcal{CN}$ represents  the complex Gaussian distributions} and $\Omega_S$ and $\Omega_I$ are the variances  of [$\mathbf{h}_{SI}]_n$ and $[\mathbf{h}_{ID}]_n$, respectively.
 All channel coefficients are assumed to be mutually  independent.
For each transmission, the received signal at  $D$ can be written  as follows
\begin{equation} \small
  y_D=\sqrt{P}\eta\mathbf{h}_{SI}^T\boldsymbol{\Phi}\mathbf{h}_{ID}x_S+w_D,
\end{equation}
where $x_S$ is the transmitted signal from $S$, $\mathbb{E}(|x_S|^2)=1$, $P$ is the transmit power,
 $w_D\sim\mathcal{CN}(0,\delta^2)$ is the  additive white Gaussian
noise at $D$,  $\boldsymbol{\Phi}=\diag(
e^{j\phi_1},e^{j\phi_2},\ldots,e^{j\phi_N})$, $\eta\in(0,1]$ is the amplitude reflection coefficient\footnote{
For simplicity, this letter assumes that all elements have the same amplitude reflection coefficient.} and
$\phi_n\in[0,2\pi]$ is the phase shift introduced by  the $n$th reflecting element.


The optimal value of $\phi_n$ is $\phi_n^*=-\arg([\mathbf{h}_{SI}]_n[\mathbf{h}_{ID}]_n)$, $n\in\{1,\ldots,N\}$.
Without considering the impact of phase errors, the full diversity order of $N$ can be achieved with this optimal setting \cite{IRS_NOMA_Ding}. However, only a finite number of quantized
 values can be considered in practice. Therefore, we consider a finite  number of
  quantization levels  $L\geq 2$, where $L$ is a positive integer.
   Accordingly, the number of quantization bits is $\log_2L$.
    The  phase shift  range  is uniformly quantized into $L$ levels, i.e.,
     $\mathcal{F}\triangleq\left\{0,\frac{2\pi}{L},\ldots,\frac{(L-1)2\pi}{L}\right\}$,
     and the phase shift of each reflecting element is designed by mapping
       its optimal value to the nearest point in $\mathcal{F}$, i.e.,
        $\phi_n=\hat{\phi}_{l_l}$, where $|\phi_n^*-\hat{\phi}_{l_l}|\leq |\phi_n^*-\hat{\phi}_{l_u}|$, 
$\hat{\phi}_{l_l},\hat{\phi}_{l_u}\in\mathcal{F}$. Then, the phase
 error $\Theta_n=\hat{\phi}_{l_l}-\phi^*_n$ for the $n$th reflecting element is a uniformly distributed random variable in $[-\frac{\pi}{L},\frac{\pi}{L}]$ \cite{IRS_Phase_Error}. Thus, the received SNR at $D$ can be written as follows
\begin{equation}\label{gamma_D}\small
  \gamma_D={\rho\eta^2\left|\sum_{n=1}^N\left|[\mathbf{h}_{SI}]_n[\mathbf{h}_{ID}]_n\right|e^{j{\Theta}_n}\right|^2}
  =\rho\eta^2\Omega_S\Omega_I{\left|\sum_{n=1}^Ng_n\right|^2},
\end{equation}
where $\rho\triangleq\frac{P}{\delta^2}$ denotes the transmit SNR and
$g_n\triangleq \frac{|[\mathbf{h}_{SI}]_n[\mathbf{h}_{ID}]_n|e^{j \Theta_n}}{\sqrt{\Omega_S\Omega_I}}$ is the normalized
channel coefficient.
{From \cite{IRS_NOMA_Ding},  the cumulative distribution function (CDF) of $|g_n|^2$ can be expressed as
   \begin{align}\label{cdf}
    \!\!\!F_{|g_n|^2}(x)\!=\!1-\!2\sqrt{x}K_1(2\sqrt{x})\thickapprox \!-x\ln x,\ \textrm{as}\ x\thickapprox 0,
  \end{align}
   where $K_1(\cdot)$ denotes the modified Bessel function of the
second kind.}

 The outage probability between $S$ and $D$ can be expressed as
follows
\begin{align}\label{definition_outage}
  P_N^{\rm out}(\rho)&=\Pr\left\{\log_2\left(1+\rho\eta^2\Omega_S\Omega_I{|G_N|^2}\right)<R_0\right\}\notag\\
  &=\Pr\left\{{|G_N|^2}<\epsilon_0\rho^{-1}\right\},
\end{align}
where $G_N\triangleq\sum_{n=1}^N g_n$ and $R_0$ denotes the target data rate in bits per channel use (bpcu)  and $\epsilon_0\triangleq
\frac{2^{R_0}-1}{\eta^2\Omega_S\Omega_I}$.
The diversity order is defined, as a function of $L$, as follows:
\begin{align}\label{definition_diversity}
  d_N(L)=\lim_{\rho\rightarrow \infty}\frac{-\log P_N^{\rm out}(\rho)}{\log \rho}.
\end{align}
\begin{Remark}
  {Although the central limit
   theorem (CLT) is widely adopted to analyze the averaged SNR of RIS-assisted systems for large values of $N$,
   it is not suitable to analyze the diversity order. In particular,  the CLT-based approximation
   has an error
floor, i.e.,  the outage probability becomes a constant when $\rho\rightarrow
\infty$ as shown in \cite{IRS_NOMA_Ding}.}
  \end{Remark}
  \begin{Remark}
{The CDF of $G_N$ is difficult to obtain, since its real part and imaginary part
are correlated. Therefore, an exact analytical expression for the outage probability $P_N^{\rm out}$ cannot be obtained, in general. Motivated by these considerations,
in the following sections, we derive  lower and upper bounds for $P_N^{\rm out}$ in order to study the achievable diversity order for different values of $L$.
}
\end{Remark}
\vspace{-1em}
\section{Diversity Order for $L=3$}\label{section_L3}
In this section, we introduce an upper bound for the outage probability in \eqref{definition_outage}
 and prove that the maximum diversity order $N$ can be achieved in this case.
When $L=3$, we have
$\Theta_n\in[-\frac{\pi}{3},\frac{\pi}{3}]$, $\forall n\in\{1,\ldots,N\}$.
 Before proceeding with the analytical derivation, we introduce two lemmas.

{
\begin{Lemma}\label{lemma_arg}
  For $\forall C_1,C_2\in \mathbb{C}$ ($\mathbb{C}$ denotes
 the  complex domain), $\arg(C_1+C_2)\in[\varphi_1,\varphi_2]$ holds
       if $\arg(C_1),\arg(C_2)\in[\varphi_1,\varphi_2]$ and $-\frac{\pi}2\leq\varphi_1\leq\varphi_2\leq \frac{\pi}2$.
\end{Lemma}
\begin{proof}
This lemma can be easily
obtained based on geometrical relationships among $\tilde{g}_i$, $\tilde{g}_j$ and $\tilde{g}_i+\tilde{g}_j$
in a complex  plane, whose details are omitted for  brevity. 
\end{proof}
}

\begin{Lemma}\label{lemma_basic}
For  $\forall b,c
 \in\mathbb{R}^+$ ($\mathbb{R}^+$ denotes the positive  real domain), 
 and $\forall n\in\{1,\ldots,N\}$, we
 have\footnote{Throughout this letter, each parameter is acquiescently
  not a function of $\rho$, i.e., its value does not change with $\rho$, unless otherwise stated.}
\begin{enumerate}
  \item[(a)]
  $\Pr\left\{(|g_n|-a)^2\!<\!b\rho^{-1}\right\}\dot\leq \rho^{-1}$,
   if 
  $0\!<\!a<\! c\rho^{- \frac12}$; 
  \item[(b)]
   $\Pr\{(|g_n|-a)^2<b\rho^{-1}\}\doteq \rho^{-\frac12}$, if 
   {$a>0$ being independent of $\rho$, i.e., $a\doteq \rho^0$.}
\end{enumerate}
where ``$\doteq$'' denotes exponential equality \cite{zheng2003diversity}, i.e., $f(\rho)\doteq \rho^{d}$ is defined as $\lim_{\rho\rightarrow\infty}\frac{\log f(\rho)}{\log \rho}=d$, and $d\in\mathbb{R}$
is the exponential order of $f(\rho)$.
 {Similarly, $f(\rho)\dot\leq\rho^{d}$ $(\dot\geq\rho^{d})$ is defined as
  $\lim_{\rho\rightarrow\infty}\frac{\log f(\rho)}{\log \rho}\leq d$ ($\geq d$).}
\end{Lemma}
\begin{proof}
See  Appendix \ref{proof_lemma_basic}.
\end{proof}

In the following two subsections, we  derive the outage probability for $N=2$ and $N \ge 2$, respectively.
 \vspace{-1em}

\subsection{$N=2$}\label{subsection(L3N2)}
When $N=2$, $G_2=g_1+g_2$. Since $|\Theta_1-\Theta_2|\leq \frac{2\pi}3$ and $G_2$ decreases with $|\Theta_1-\Theta_2|$, $|G_2|^2$ can be lower bounded
\begin{align}\label{G_2(L=3)lower}
  |G_2|^2&=|g_2|^2+|g_1|^2+2|g_2||g_1|\cos{(|\Theta_2-\Theta_1|)}\notag\\
  &\geq |g_2|^2+|g_1|^2+2|g_2||g_1|\cos\left(\frac{2\pi}3\right)\notag\\
  &=\left(|g_2|-\frac12 |g_1|\right)^2+\frac34 |g_1|^2.
\end{align}
Thus, the outage probability can be upper bounded
 \begin{align}
     &P_2^{\rm out}(\rho)\leq \Pr\left\{\left(|g_2|-\frac12 |g_1|\right)^2+\frac34 |g_1|^2<\epsilon_0\rho^{-1}\right\}\notag\\
     &\stackrel{(a)}{<} \Pr\left\{\left(|g_2|-\frac12 |g_1|\right)^2<\epsilon_0\rho^{-1},\ \frac34 |g_1|^2<\epsilon_0\rho^{-1}\right\}\notag\\
     &=\Pr\left\{|g_1|^2<\omega_1\rho^{-1}\right\}\notag\\
     &\quad \cdot \Pr\left\{\left(|g_2|-\frac12 |g_1|\right)^2<\epsilon_0\rho^{-1}\bigg| |g_1|<\sqrt{\omega_1}\rho^{-\frac12}\right\},
     \label{proof_lem_eq1}
 \end{align}
 where $\omega_1\triangleq \frac43\epsilon_0$ and $(a)$ follows from $\max(x,y)<x+y$ for $x,y>0$. From
 \eqref{cdf}, as $\rho\rightarrow\infty$,  we obtain \begin{align}\label{proof_lem_eq2}
 \Pr\left\{|g_1|^2<\omega_1\rho^{-1}\right\}\thickapprox -\omega_1\rho^{-1}\ln \left(\omega_1\rho\right)\doteq \rho^{-1}.\end{align}
 Furthermore, by setting $g_n=g_2$, $a=\frac12|g_1|$, $b=\epsilon_0$ and $c=\frac12\sqrt{\omega_1}$
 in Lemma \ref{lemma_basic}-(a), we have
\begin{align}  \label{proof_le_eq2}
    \!\Pr\left\{\left(|g_2|\!-\!\frac12 |g_1|\right)^2\!<\!\epsilon_0\rho^{-1}\bigg| |g_1|\!<\!\sqrt{\omega_1}\rho^{-\frac12}\right\}
    \dot\leq \rho^{-1}.
\end{align}
Substituting \eqref{proof_lem_eq2} and \eqref{proof_le_eq2} into \eqref{proof_lem_eq1}, we obtain
 $P_2^{\rm out}(\rho)\dot\leq \rho^{-2} $. Therefore, we evince that
{   $d_2(3)\geq 2$. Since the maximum  diversity order is $2$ when $N=2$, we conclude that $d_2(3)= 2$.}
\vspace{-1em}
\subsection{$N\geq 2$}
The  result  for  $N=2$ can be generalized to the  case  $N\geq 2$, as stated in
 the following proposition. 
 \begin{Proposition}\label{proposition(L=3N>2)}
 For $\forall \lambda\in\mathbb{R}^+$ and $\forall  N\geq2$, we have
 \begin{align}\label{proposition_upper}
   \Pr\left\{|G_{N}|^2<\lambda\rho^{-1}\right\}\dot\leq \rho^{-N}.
   \end{align}
 \end{Proposition}
 \begin{proof}
This proposition can be proved using mathematical induction. We consider the following induction steps.
\subsubsection{  $N=2$}
{ This  case is considered in Section \ref{subsection(L3N2)},
and \eqref{proposition_upper} holds
  by replacing $\epsilon_0$ with $\lambda$.}
\subsubsection{ $N=k-1$}
When $N=k-1$, $k\geq 3$, we assume that \eqref{proposition_upper}  holds true, i.e.,
\begin{align}\label{assum1}
   \Pr\left\{|G_{k-1}|^2<\lambda\rho^{-1}\right\}\dot\leq \rho^{-(k-1)}, \forall
   \lambda\in\mathbb{R}^+,
   \end{align}
   where $G_{k-1}=\sum_{n=1}^{k-1} g_n.$ This is the induction hypothesis in our problem formulation.

\subsubsection{  $N=k$}
Based on the induction hypothesis corresponding to $N=k-1$,
we proceed as follows.
Based on Lemma \ref{lemma_arg},
we have $\arg(G_{k-1})\in \left[-\frac\pi3,\frac\pi3\right]$, and hence
 $|G_{k}|^2=|g_{k}+G_{k-1}|^2$ can be bounded as follows
\begin{align}
  |G_k|^2\geq \left(|g_k|-\frac12 |G_{k-1}|\right)^2+\frac34 |G_{k-1}|^2,
\end{align}
with the aid of steps similar to \eqref{G_2(L=3)lower}.
In addition, similar to \eqref{proof_lem_eq1}, for $\forall \lambda\in\mathbb{R}^{+}$, we can obtain
\begin{align}\label{proof_P2_step3}
     &\!\!\!\Pr\left\{|G_{k}|^2<\lambda\rho^{-1}\right\}\leq \Pr\left\{|G_{k-1}|^2<\frac43\lambda\rho^{-1}\right\} \cdot\Pr\bigg\{\notag\\
     &\left.\left(|g_k|\!-\!\frac12 |G_{k-1}|\right)^2\!\!<\!\lambda\rho^{-1}\bigg| |G_{k-1}|\!<\!\left({\frac43\lambda}\right)^{\frac12}\rho^{-\frac12}\right\}.
 \end{align}
Finally, by applying  Lemma \ref{lemma_basic}-(a) and \eqref{assum1} to \eqref{proof_P2_step3},
we obtain $\Pr\left\{|G_{k}|^2<\lambda\rho^{-1}\right\}\dot\leq \rho^{-k}$  for $\forall\lambda\in\mathbb{R}^{+}$. This completes the proof.
 \end{proof}

 Based on Proposition \ref{proposition(L=3N>2)}, $P_N^{\rm out}(\rho)\dot\leq \rho^{-N}$,
{i.e., $d_N(3)\geq N$}  holds for
 $\forall N\geq2$. {Since the achievable diversity order cannot exceed $N$,
  we have
  the following theorem.}
\begin{theorem}\label{theoremL3}
If the number of  quantization levels is  $L\geq 3$,
the full diversity order of $N$  can be achieved, i.e.,  $ d_N(L)=N$, $\forall L\geq 3,\  N\geq 2$.
 \end{theorem}

 Now, one may wonder  whether $L=3$ is the minimum required number of phase quantization levels for achieving
  the full diversity order. To answer this question, we analyze the setup for  $L=2$ in the next section.
\vspace{-1em}
\section{Diversity Order for $L=2$}\label{section_L2}
This section aims to prove that the full diversity of $N$ cannot be achieved if only two
quantization levels are used.
If $L\!=\!2$,  in particular, we have
$\mathcal{F}\!=\!\{0,\pi\}$ and $\Theta_n\in[-\frac{\pi}{2},\frac{\pi}{2}]$, $\forall n\!\in\!\{1,\ldots,N\}$.
To prove this result,  we first compute a lower bound for the outage probability, and prove that the corresponding
 exponential order is
$-\frac{N\!+\!1}{2}$.  Based on this result, we conclude that the diversity order cannot exceed     $\frac{N\!+\!1}{2}$.


\vspace{-1em}
\subsection{$N=2$}
As mentioned, 
we need to identify  a
 lower bound for the outage probability. 
 For fixed $|g_1|$ and $|g_2|$,   {$|G_2|$ attains   its minimum  value} if $|\Theta_1-\Theta_2|=\pi$.
Therefore, we are interested in computing  the outage probability conditioned on the event that
$|\Theta_1-\Theta_2|$ is close to $\pi$.
This is provided in  the following proposition.  
\begin{Proposition}\label{lemma_pi}
At high SNR, define the event $\varepsilon_1$ as   follows \begin{align}&\varepsilon_1\triangleq\Big\{\Theta_i\in\Big[-\frac{\pi}{2},-\frac{\pi}{2}+\theta \Big], \ \Theta_j\in\Big[\frac{\pi}{2}-\theta,\frac{\pi}{2} \Big],\notag\\
&\hspace{3cm} (i,j)=(1,2)\ \textrm{or}\ (2,1)\Big\},\end{align}
where $\theta\triangleq \rho^{-\frac12}$. Conditioned on the event $\varepsilon_1$,
 we have  \begin{align}
 \Pr\left\{|G_2|^2\!<\!\epsilon_0\rho^{-1}|\varepsilon_1\right\}
   &\dot\geq \rho^{-\frac12}.
 \end{align}
\end{Proposition}

\begin{proof}
Without loss of generality for the event $\varepsilon_1$, we assume $\Theta_1\in\big[-\frac{\pi}{2},-\frac{\pi}{2}+\theta \big], \ \Theta_2\in\big[\frac{\pi}{2}-\theta,\frac{\pi}{2} \big]$. Hence  $\pi\!-\!2\theta\leq\Theta_2\!-\!\Theta_1\!\leq\! \pi$ and $|G_2|^2$ can be upper bounded as follows
  \begin{align}\label{proof_lemma_step1}
  &|G_2|^2
  =
  |g_1|^2+|g_2|^2+2|g_1||g_2|\cos{(\Theta_2-\Theta_1)}\notag\\
  &\leq |g_1|^2+|g_2|^2-2\cos(2\theta)|g_1||g_2|\notag\\
  &=(|g_1|-\cos(2\theta)|g_2|)^2+\sin^2(2\theta)|g_2|^2\notag\\
  &<\! (|g_1|-\cos(2\theta) |g_2|)^2\!+4\rho^{-1}|g_2|^2\!\triangleq \Lambda\left(|g_1|^2,|g_2|^2\right),
  \end{align}
  where the last step is based on the definition $\theta=\rho^{-\frac12}$ and
   the fact that $\sin(\theta)<\theta$.
Furthermore, 
  we have
  \begin{align}
    &\quad\Pr\left\{|G_2|^2<\epsilon_0\rho^{-1}|\varepsilon_1\right\}\notag\\
    &>
    \Pr\left\{\Lambda\left(|g_1|^2,|g_2|^2\right)<\epsilon_0\rho^{-1}\right\}\notag\\
    &{\stackrel{(a)}{>}
    \Pr\Big\{\Lambda\left(|g_1|^2,|g_2|^2\right)<\epsilon_0\rho^{-1},\ {\omega_2}\leq |g_2|^2\leq {2\omega_2}\Big\}}\notag\\
    &=\Pr\left\{{\omega_2}\leq |g_2|^2\leq{2\omega_2}\right\}\cdot
    \Pr\Big\{\Lambda\left(|g_1|^2,|g_2|^2\right)
    \notag\\
    &\qquad\qquad\qquad\qquad\qquad\left.<\epsilon_0\rho^{-1}\Big | \sqrt{\omega_2}\leq |g_2|\leq \sqrt{2\omega_2}\right\}\notag\\
    &\stackrel{(b)}{>} \left[F_{|g_2|^2}\left({2\omega_2}\right)-F_{|g_2|^2}\left({\omega_2}\right)\right]\cdot
    \Pr\Big\{\big(|g_1|\notag\\
    &\qquad \quad -\cos(2\theta) |g_2|)^2<\left.\frac{\epsilon_0\rho^{-1}}{2}\bigg |
    \sqrt{\omega_2}\leq |g_2|\leq\sqrt{2\omega_2}\right\}\notag\\
    &\doteq \rho^{-\frac12},
      \label{proof_lemma_step2}
    \end{align}
    where $\omega_2\triangleq \frac{\epsilon_0}{16}$; {$(a)$ holds since $\Pr\{A\}>\Pr\{A,B\}$
    for  all events $A$ and $B$ satisfying $B\nsubseteq A$,} and is used   to
    limit the range of $|g_2|$ to $[\sqrt{\omega_2},\sqrt{2\omega_2}]$ so that
$|g_2|\doteq \rho^0$  holds for $ \forall |g_2|\in[\sqrt{\omega_2},\sqrt{2\omega_2}]$  and Lemma \ref{lemma_basic}-(b) can be applied in the last step; $(b)$ is due to
the fact that $4\rho^{-1}|g_2|^2<\frac{\epsilon_0\rho^{-1}}2$ if $|g_2|<\sqrt{2\omega_2}$; and the last step is based on Lemma \ref{lemma_basic}-(b)
 (setting $a=\cos(2\theta) |g_2|$ and $b=\epsilon_0/2$).
\end{proof}
\begin{Remark}
The motivation for setting  $\theta=\rho^{-\frac12}$ is to make sure that
 the term $4\theta^2 |g_2|^2$  in \eqref{proof_lemma_step1} satisfies   $4\theta^2 |g_2|^2<\epsilon_0\rho^{-1}$ even if $|g_2|^2\doteq \rho^0$. In this case, $|g_2|^2$
 can take any values in  the range  $[\omega_2,2\omega_2]$, so that $|g_2|^2\doteq \rho^0$
is satisfied and 
the conditional outage probability $\Pr\left\{|G_2|^2<\epsilon_0\rho^{-1}|\varepsilon_1\right\}$
 can be
exponentially lower bounded by $\rho^{-\frac12}$ as stated in \eqref{proof_lemma_step2} . 
\end{Remark}

Moreover, the probability of the event $\varepsilon_1$ can be computed as
$\Pr\{\varepsilon_1\}=2\left(\frac{\theta}{\pi}\right)^2\doteq \rho^{-1}$, since
\begin{align}\label{probabi}
\!\!\Pr\Big\{\!\Theta_i\!\in\!\Big[-\!\frac{\pi}{2},-\!\frac{\pi}{2}\!+\!\theta \Big]\!\Big\}
\!=\!\Pr\Big\{\!\Theta_i\!\in\!\Big[\frac{\pi}{2},\frac{\pi}{2}\!-\!\theta \Big]\!\Big\}\!=\!\frac\theta\pi,
\end{align} where $i=1,2$.
Therefore, based on Proposition \ref{lemma_pi},
the outage probability can be exponentially lower bounded as follows
\begin{align}\label{Po(L=2N=2)}
  &\quad P_2^{\rm out}(\rho)\geq \Pr\{\varepsilon_1\}\cdot\Pr\left\{|G_2|^2<\epsilon_0\rho^{-1}|\varepsilon_1\right\}
  \dot\geq \rho^{-\frac32}.
\end{align}
Thus, $d_N(2)\leq \frac32$, which is smaller than the theoretical full diversity order of $N=2$.
\vspace{-1em}
\subsection{ $N\geq 2$}
Proposition  \ref{lemma_pi} can be generalized  to the case  $N\geq 2$,
 as shown in  the following corollary.
\begin{Corollary}\label{lemma_pi2}
At high SNR, define the event $\varepsilon_2$ as  follows \begin{align}&\varepsilon_2\triangleq\Big\{\Theta_i\in\Big[-\frac{\pi}{2},-\frac{\pi}{2}+\theta \Big], \ \{\Theta_j\}_{j\in\{1,\ldots,N\},j\neq i}\in\notag\\
 &\hspace{3cm}\Big[\frac{\pi}{2}-\theta,\frac{\pi}{2} \Big],\ \forall {i\in\{1,\ldots, N\}}\Big\}.\end{align}
Conditioned on the event $\varepsilon_2$,
 we have  $\Pr\left\{|G_N|^2<\epsilon_0\rho^{-1}|\varepsilon_2\right\}\dot\geq \rho^{-\frac12}$.
\end{Corollary}
\begin{proof}
For the event $\varepsilon_2$, assume $\Theta_1\in\big[-\frac{\pi}{2},-\frac{\pi}{2}+\theta \big]$
and $ \{\Theta_j\}_{j\in\{2,\ldots,N\}}\in
 \big[\frac{\pi}{2}-\theta,\frac{\pi}{2} \big]$ without loss of generality.
Moreover,
define $\tilde{G}_{N-1}\triangleq \sum_{j=2}^N g_j$, and hence $G_N=g_1+\tilde{G}_{N-1}$.
According to Lemma \ref{lemma_arg}, we obtain
$\arg (\tilde{G}_{N-1})\in\big[\frac{\pi}{2}-\theta,\frac{\pi}{2} \big]$.
Therefore, by using a similar analytical derivation as for Proposition  \ref{lemma_pi}, 
 the corollary  can be proved by replacing $g_2$ with $\tilde{G}_{N-1}$ in \eqref{proof_lemma_step1}
 and \eqref{proof_lemma_step2}.
\end{proof}

In addition, from \eqref{probabi}, the probability of the event $\varepsilon_2$ is $\Pr\{\varepsilon_2\}=N\left(\frac{\theta}{\pi}\right)^N\doteq \rho^{-\frac N2}$.
Therefore, we obtain
\begin{align}\label{Po(L=2N>2)}
  P_N^{\rm out}(\rho)\geq \Pr\{\varepsilon_2\}\cdot\Pr\left\{|G_N|^2<\epsilon_0\rho^{-1}|\varepsilon_2\right\}
  \dot\geq \rho^{-\frac {N+1}2}.
\end{align}

From \eqref{Po(L=2N>2)}, the upper bound of the diversity order for $L=2$
 can be summarized in the following theorem.
\begin{theorem}\label{theoremL2}
If the number of quantization levels is  $L=2$,
the diversity order   is upper bounded by
  $d_N(2)\leq \frac{N+1}2$, $\forall N\geq 2$.
\end{theorem}
\begin{Remark}
Combing Theorem \ref{theoremL3} and Theorem \ref{theoremL2}, the minimum required number of
 phase quantization levels to
  achieve the full diversity order of $N$ is equal to $L=3$.
\end{Remark}

\begin{figure}[tbp]
 \begin{center}
{
\includegraphics[width=0.85\columnwidth]{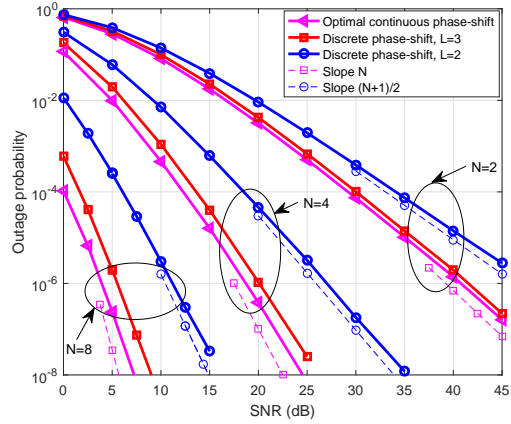}}
\end{center}
\vspace*{-1.5em}
\caption{\small{The outage probability  versus the  SNR, for different values of $L$ and $N$, where $R_0=1$ bpcu}.}\label{Outagesnr}\vspace{-0.5em}
\end{figure}
\vspace{-1em}

\section{Impact of Direct Link}\label{section_discussion}
If the direct link  between $S$ and $D$ cannot be ignored,  Theorem \ref{theoremL3} and Theorem \ref{theoremL2} can be modified as follows.

Let $h_{SD}$ denote the channel of the direct link. The optimal value of $\phi_n$ in Section \ref{section_model} is
 $\phi_n^*=\arg(h_{SD})-\arg([\mathbf{h}_{SI}]_n[\mathbf{h}_{ID}]_n)$, $n\in\{1,\ldots,N\}$. Moreover, the SNR is $\gamma_D=\rho\Big||h_{SD}|+\eta\sqrt{\Omega_S\Omega_I}G_N\Big|^2$. Based on
Lemma \ref{lemma_arg},  $|\arg(G_N)|$ cannot exceed $\frac{\pi}2$ even if $ L= 2$, so $\gamma_D\geq \rho(|h_{SD}|^2+\eta^2{\Omega_S\Omega_I}|G_N|^2)$ for $\forall L\geq 2$. Based on this lower bound on $\gamma_D$, Theorem \ref{theoremL3} still holds by replacing the diversity order  $N$  with $N+1$.

Furthermore, $\gamma_D$ can be upper bounded as $\gamma_D\leq \rho(|h_{SD}|+\eta\sqrt{\Omega_S\Omega_I}|G_N|)^2$.
Therefore,   Theorem \ref{theoremL2}  holds by replacing the upper bound  $\frac{N+1}2$ with $\frac{N+3}2$. The details are omitted due to space limitations.
\vspace{-1em}
\section{Numerical Results}\label{section numerical}
In this section, Monte Carlo simulations are provided to verify the accuracy of the analytical results about
the diversity order. The direct link between $S$ and $D$ does not exist, and we set $[\mathbf{h}_{SI}]_n\sim\mathcal{CN}(0,1)$ and $[\mathbf{h}_{ID}]_n\sim\mathcal{CN}(0,0.5)$, $\forall
n\in\{1,\ldots,N\}$, i.e., $\Omega_S=1$ and $\Omega_I=0.5$. We also set   $\eta=0.8$.
\begin{figure}[tbp]
 \begin{center}
{
\includegraphics[width=3in,height=2.38in]{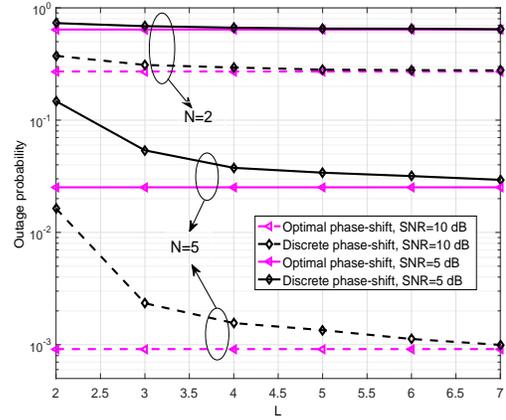}}
\end{center}
\vspace{-1.5em}
\caption{\small{The outage probability versus $L$,
for different values of the transmit SNR and $N$, where $R_0=2$ bpcu}.}\label{OutageL}\vspace{-0.5em}
\end{figure}


The effect of phase quantization  on the diversity order of RIS-aided transmission is illustrated in Fig. \ref{Outagesnr}.
{For ease of interpretation, we reports the lines  with slopes $\frac{N+1}2$ and $N$ as well.
 When $L=2$, for $N=2$, $N=4$ and $N=8$, we can  observe that the slope of the outage probability
curves is not smaller than $-1.5$, $-2.5$ and $-4.5$ (i.e.,
 the diversity order does not exceed  $1.5$, $2.5$ and $4.5$) at high SNRs, respectively.}
This observation is consistent with  Theorem \ref{theoremL2}.
The main reason is that the phase error is  uniformly distributed over
 $[- \frac{\pi}{2}, \frac{\pi}{2}]$, and there exists a  probability that
  each phase error is close to the boundary of $- \frac{\pi}{2}$ or $ \frac{\pi}{2}$.
  In this case, the amplitudes of the channel coefficients  mutually cancel out,
  which results in a loss of the diversity order.
When $L\geq 3$, on the other hand, we  observe  that the corresponding  curves are parallel to the curves
 obtained in the absence of phase errors. We conclude that the full diversity order can be achieved as long as $L\geq3$, which is consistent with Theorem \ref{theoremL3}.


In Fig. \ref{OutageL}, the outage probability  as a function of the
phase quantization level $L$ is shown. We note that the outage performance is a decreasing function of  $L$.  {As  $L$ increases, the outage probability  approaches  the lower bound with perfect phase shifts.}
For both setups  $N=2$ and $N=5$, the outage probability has the largest values if $L=2$. When $N=2$, the outage probability  approaches the lower bound with perfect phase shifts provided that $L\geq 3$. When $N=5$, the gap between the curve with phase errors and the curve with perfect phase shifts is larger. However, the gap reduces as $L$ increases. Therefore, the  quantization errors for phase shifts of the RIS do not lead to significant  outage performance loss, which is a promising finding for the deployment of RISs in wireless systems.
\vspace{-0.8em}
 \section{Conclusion}
This letter investigated the diversity order of RIS-aided communication systems with discrete phase shifts.
 The main contribution of this letter was to unveil  the minimum number of phase quantization levels to achieve  full diversity. In particular, the full diversity order is proved to be achievable, if and only if at least three quantization levels are used. Simulation results verified  the theoretical finding and showed that the outage performance loss is negligible even for moderate values of the quantization levels.




%
\vspace{-1em}
\appendices
\section{Proof of Lemma \ref{lemma_basic}}\label{proof_lemma_basic}
We consider different values of $a$ as follows.
 \subsubsection{For $0<a< c\rho^{-\frac12}$} We have
  \begin{align}&\quad \Pr\{(|g_n|-a)^2<b\rho^{-1}\}\notag\\
  &=\Pr\left\{
  -\sqrt{b}\rho^{-\frac12}+a<|g_n|<\sqrt{b}\rho^{-\frac12}+a\right\}\label{step1}\\
  &\leq \Pr\left\{|g_n|^2<\left(\sqrt{b}\rho^{-\frac12}+c\rho^{-\frac12}\right)^2\right\}\notag\\
  &=1- 2\left(\sqrt{b}+c\right)\rho^{-\frac12}
  K_1\left(2\left(\sqrt{b}+c\right)\rho^{-\frac12}\right)\notag\\
  &\thickapprox -\left(\sqrt{b}+c\right)^2\rho^{-1}
  \ln\left(\left(\sqrt{b}+c\right)^2\rho^{-1}\right),
  \end{align}
  as $\rho\rightarrow \infty$, where the last step is based on \eqref{cdf}.
  Thus, $ \Pr\{(|g_n|-a)^2<b\rho^{-1}\}\dot\leq \rho^{-1}$.

 \subsubsection{For $a>0$ and $a\doteq \rho^0$} Based on \eqref{step1}, we have
    \begin{align}&\quad \Pr\{(|g_n|-a)^2<b\rho^{-1}\}\notag\\
  &=\Pr\left\{
  \left(-\sqrt{b}\rho^{-\frac12}+a\right)^2<|g_n|^2<\left(\sqrt{b}\rho^{-\frac12}+a\right)^2\right\}\notag\\
  &=
  \tilde{K}_1\left(2a-z\right)-\tilde{K}_1\left(2a+z\right), \label{step2}
  \end{align}
  where $\tilde{K}_1(x)\triangleq xK_1(x)$, $z\triangleq2\sqrt{b}\rho^{-\frac12}$ and $2a\gg z$ holds as $\rho\rightarrow \infty$. From \cite[Eq. 8.446]{gradsh2000table}, the function $\tilde{K}_1(x)$
  has the following series representation:
  \begin{align}
    \tilde{K}_1(x)=1-\sum_{k=0}^{\infty}\frac{\left(\frac{x}{2}\right)^{2k+2}\left(A(k)-\ln\frac{x}{2}
    \right)}{k!(k+1)!},
  \end{align}
  where $A(k)\triangleq\frac12 \psi(k+1)+\frac12 \psi(k+2)$, and the function $\psi(\cdot)$ is defined
  in \cite[Eq. 8.365.3]{gradsh2000table}.
  Therefore, by using the Taylor expansion at the point $2a$,  we have
  \begin{align}\label{appendix_eq1}
    &\tilde{K}_1\left(2a-z\right)=\tilde{K}_1(2a)-\tilde{K}'_1(2a)z+o(z),
  \end{align}
  where $o(\cdot)$ denotes higher order terms and $\tilde{K}'_1(2a)$ can be expressed as follows
  \begin{align}
   \tilde{K}'_1(2a)= -\sum_{k=0}^{\infty}\frac{a^{2k+1}\left[(k+1)(A(k)-\ln a)+\frac12\right]}
    {k!(k+1)!}.
  \end{align}
  Similarly, $\tilde{K}_1\left(2a+z\right)$ can be expressed as
  \begin{align}\label{appendix_eq2}
    &\tilde{K}_1\left(2a+z\right)=\tilde{K}_1(2a)+\tilde{K}'_1(2a)z+o(z).
  \end{align}
Inserting \eqref{appendix_eq1} and \eqref{appendix_eq2} into \eqref{step2}, we have
  \begin{align} \Pr\{(|g_n|-a)^2<b\rho^{-1}\}
  \thickapprox - 4\sqrt{b}\rho^{-\frac12}\tilde{K}'_1(2a)\doteq \rho^{-\frac12}.  \end{align}


\vspace{-0.8em}

\bibliographystyle{IEEEtran}
\bibliography{refference}

\end{document}